\documentclass[conference,a4paper]{IEEEtran}

\usepackage[utf8]{inputenc} 
\usepackage[T1]{fontenc}
\usepackage{url}
\usepackage{ifthen}
\usepackage{cite}
\usepackage[cmex10]{amsmath} 

\usepackage{enumitem}
\usepackage{algorithm, algorithmic, float}
\usepackage[font=small]{caption}
\usepackage[dvips]{graphicx}
\usepackage{multirow}
\usepackage{colortbl}
\usepackage{tabularx}
\usepackage{color}
\usepackage{setspace}
\usepackage{mathtools} 
\usepackage{bm}
\usepackage{amsmath,amsfonts,amssymb,amsthm,epsfig,epstopdf,url,array}
\usepackage{dsfont}
\usepackage{subfig}
\usepackage{float}
\DeclarePairedDelimiter{\floor}{\lfloor}{\rfloor}
\DeclarePairedDelimiter{\ceil}{\lceil}{\rceil}

\newcommand{\divides}{\mid}
\newcommand{\notdivides}{\nmid}
\newtheoremstyle{dotless}{}{}{\itshape}{}{\bfseries}{}{ }{}

\begin{document}
\title{A Class of MSR Codes \\ for Clustered Distributed Storage} 


\author{%
  \IEEEauthorblockN{Jy-yong Sohn, Beongjun Choi and Jaekyun Moon}
  \IEEEauthorblockA{KAIST\\
                    School of Electrical Engineering\\
                    Email: \{jysohn1108, bbzang10\}@kaist.ac.kr, jmoon@kaist.edu}
}


\maketitle

\begin{abstract}
Clustered distributed storage models real data centers where intra- and cross-cluster repair bandwidths are different.
In this paper, exact-repair minimum-storage-regenerating (MSR) codes achieving capacity of clustered distributed storage are designed. 
Focus is given on two cases: $\epsilon=0$ and $\epsilon=1/(n-k)$, where $\epsilon$ is the ratio of the available cross- and intra-cluster repair bandwidths, $n$ is the total number of distributed nodes and $k$ is the number of contact nodes in data retrieval. The former represents the scenario where cross-cluster communication is not allowed, while the latter corresponds to the case of minimum cross-cluster bandwidth that is possible under the minimum storage overhead constraint. 
For the $\epsilon=0$ case, two types of locally repairable codes are proven to achieve the MSR point. As for $\epsilon=1/(n-k)$, an explicit MSR coding scheme is suggested for the two-cluster situation under the specific of condition of $n=2k$.

\end{abstract}


\section{Introduction}

Distributed Storage Systems (DSSs) have been deployed by various enterprises to reliably store massive amounts of data under the frequent storage node failure events. A failed node is regenerated (repaired) by collecting information from other survived nodes with the regeneration process guided by a pre-defined network coding scheme. Under this setting, Dimakis \textit{et al.} \cite{dimakis2010network} obtained the expression for the maximum reliably storable file size, denoted as \textit{capacity} $\mathcal{C}(\alpha, \gamma)$, as a function of given system parameters: the node capacity $\alpha$ and the bandwidth $\gamma$ required for repairing a failed node. 
The capacity analysis in \cite{dimakis2010network} underscores the following key messages. First, there exists a network coding scheme which utilizes the $(\alpha,\gamma)$ resources and enables a reliable storage of a file of size $\mathcal{C}(\alpha, \gamma)$. Second, it is not feasible to find a network coding scheme which can reliably store a file larger than $\mathcal{C}(\alpha, \gamma)$, given the available resources of $(\alpha, \gamma)$. In subsequent research efforts, the authors of \cite{rashmi2009explicit, cadambe2013asymptotic, ernvall2014codes} proposed explicit network coding schemes which achieve the capacity of DSSs. These coding schemes are optimal in the sense of efficiently utilizing $(\alpha, \gamma)$ resources for maintaining the reliable storage systems.

Focus on the clustered nature of distributed storage has been a recent research direction taken by several researchers \cite{sohn2016capacity, sohn2017TIT, prakash2017storage, hu2017optimal}.
According to these recent papers, storage nodes dispersed into multiple \textit{racks} in real data centers are 
seen as forming \textit{clusters}. 
In particular, the authors of the present paper proposed a system model for clustered DSSs in \cite{sohn2016capacity} that reflects the difference between intra- and cross-cluster bandwidths. 
In the system model of \cite{sohn2016capacity}, the file to be stored is coded and distributed into $n$ storage nodes, which are evenly dispersed into $L$ clusters. Each node has storage capacity of $\alpha$, and the data collector contacts arbitrary $k$ out of $n$ existing nodes to retrieve the file. 
Since nodes are dispersed into multiple clusters, the regeneration process involves utilization of both intra- and cross-cluster repair bandwidths, denoted by $\beta_I$ and $\beta_c$, respectively.
In this proposed system model, the authors of \cite{sohn2016capacity} obtained the closed-form expression for the maximum reliably storable file size, or \textit{capacity} $\mathcal{C}(\alpha, \beta_I, \beta_c)$, of the clustered DSS. Furthermore, it has been shown that network coding exists that can achieve the capacity of clustered DSSs. However, explicit constructions of capacity-achieving network coding schemes for clustered DSSs have yet to be found. 

This paper proposes a network coding scheme which achieves capacity of the clustered DSS, with a minimum required node storage overhead. In other words, the suggested code is shown to be a minimum-storage-regenerating (MSR) code of the clustered DSS. 
This paper focuses on two important cases of $\epsilon=0$ and $\epsilon=1/(n-k)$, where $\epsilon\coloneqq\beta_c/\beta_I$ represents the ratio of cross- to intra-cluster repair bandwidths. The former represents the system where cross-cluster communication is not possible. The latter corresponds to the minimum $\epsilon$ value that can achieve the minimum storage overhead of $\alpha = \mathcal{M}/k$, where $\mathcal{M}$ is the file size.
When $\epsilon=0$, it is shown that appropriate application of locally repairable codes suggested in \cite{papailiopoulos2014locally,tamo2016optimal} achieves the MSR point for general $n,k,L$ settings with the application rule depending on the parameter setting. 
For the $\epsilon=1/(n-k)$ case, an explicit coding scheme is suggested which is proven to be an MSR code under the conditions of $L=2$ and $n=2k$.
There have been some previous works \cite{tebbi2014code, hu2017optimal, sahraei2017increasing, prakash2017storage} on code construction for DSS with clustered storage nodes, but to a limited extent. The works of \cite{hu2017optimal,tebbi2014code} suggested a coding scheme which can reduce the cross-cluster repair bandwidth, but these schemes are not proven to be an MSR code that achieves capacity of clustered DSSs with minimum storage overhead. 
The authors of \cite{sahraei2017increasing} provided an explicit coding scheme which reduces the repair bandwidth of a clustered DSS 
under the condition that each failed node can be exactly regenerated by contacting any one of other clusters. 
However, the approach of \cite{sahraei2017increasing} is different from that of the present paper in the sense that it does not consider the scenario with unequal intra- and cross-cluster repair bandwidths.
Moreover, the coding scheme proposed in \cite{sahraei2017increasing} is shown to be a minimum-bandwidth-regenerating (MBR) code for some limited parameter setting, while the present paper deals with an MSR code.
An MSR code for clustered DSSs has been suggested in \cite{prakash2017storage}, but this paper has the data retrieval condition different from the present paper. The authors of \cite{prakash2017storage} considered the scenario where data can be collected by contacting arbitrary $k$ out of $n$ clusters, while data can be retrieved by contacting arbitrary $k$ out of $n$ nodes in the present paper. Thus, the two models have the identical condition only when each cluster has one node. The difference in data retrieval conditions results in different capacity values and different MSR points. In short, the code in \cite{prakash2017storage} and the code in this paper achieves different MSR points.

\newtheorem{theorem}{Theorem}
\newtheorem{lemma}{Lemma}
\newtheorem{corollary}{Corollary}
\newtheorem{definition}{Definition}
\newtheorem{prop}{Proposition}
\newtheorem{construction}{Construction}
\newtheorem{remark}{Remark}

\theoremstyle{dotless}
\newtheorem{thm}{Theorem}
\newtheorem{cor}{Corollary}

\section{Backgrounds and Notations}




A given file of $\mathcal{M}$ symbols is encoded and distributed into $n$ nodes, each of which has node capacity $\alpha$. The storage nodes are evenly distributed into $L \geq 2$ clusters, so that each cluster contains $n_I\coloneqq n/L$ nodes. A failed node is regenerated by
 obtaining information from other survived nodes: $n_I-1$ nodes in the same cluster help by sending $\beta_I$ each, while $n-n_I$ nodes in other clusters help by sending $\beta_c$ each.
 Thus, repairing each node requires the overall repair bandwidth of
 \begin{equation} \label{Eqn:gamma}
 \gamma = (n_I-1)\beta_I + (n-n_I)\beta_c.
 \end{equation}
 A data collector (DC) retrieves the original file $\mathcal{M}$ by contacting arbitrary $k$ out of $n$ nodes - this property is called the maximum-distance-separable (MDS) property. 
The clustered distributed storage system 
with parameters $n,k,L$ is called an $[n,k,L]$-clustered DSS.
In an $[n,k,L]$-clustered DSS with given parameters of $\alpha, \beta_I, \beta_c$, \textit{capacity} $\mathcal{C}(\alpha, \gamma)$ is defined in \cite{sohn2016capacity} as the maximum data that can be reliably stored. The closed-form expression for $\mathcal{C}(\alpha, \gamma)$ is obtained in Theorem 1 of \cite{sohn2016capacity}. 
Aiming at reliably storing file $\mathcal{M}$, the set of $(\alpha, \gamma)$ pair values is said to be \textit{feasible} if $\mathcal{C}(\alpha, \gamma) \geq \mathcal{M}$ holds. According to Corollaries 1 and 2 of \cite{sohn2017TIT}, the set of feasible $(\alpha, \gamma)$ points shows the optimal trade-off relationship between $\alpha$ and $\gamma$, as illustrated in Fig. \ref{Fig:MBR_MSR_points}.
In the optimal trade-off curve, the point with minimum node capacity $\alpha$ is called the minimum-storage-regenerating (MSR) point. Explicit regenerating codes that achieve the MSR point are called the MSR codes. 
According to Theorem 3 of \cite{sohn2017TIT}, node capacity of the MSR point satisfies
\begin{align}
\alpha_\text{msr} &= \mathcal{M}/k \quad \quad \text{if  } \epsilon \geq \frac{1}{n-k}, \label{Eqn:alpha_MSR_large_epsilon_val} \\
\alpha_\text{msr} &> \mathcal{M}/k \quad \quad \text{if  } 0 \leq \epsilon < \frac{1}{n-k}. \label{Eqn:alpha_MSR_small_epsilon_val}
\end{align}
Note that $\alpha = \mathcal{M}/k$ is the minimum storage overhead to satisfy the MDS property, as stated in \cite{dimakis2010network}. Thus, $\epsilon=1/(n-k)$ is the scenario with minimum cross-cluster communication when the minimum storage overhead constraint $\alpha=\mathcal{M}/k$ is imposed.

Here we introduce some useful notations used in the paper. For a positive integer $n$, 
$[n]$ represents the set 
$\{1,2,\cdots, n\}$. 
For natural numbers $a$ and $b$, we use the notation $a \divides b$ if $a$ divides $b$. Similarly, write $a \notdivides b$ if $a$ does not divide $b$.
For given $k$ and $n_I$, 
we define
\begin{align}
q & \coloneqq \floor{\frac{k}{n_I}}, \label{Eqn:quotient} \\
m &\coloneqq mod(k,n_I) = k-qn_I. \label{Eqn:remainder}
\end{align}
For vectors we use bold-faced lower case letters. For a given vector 
$\mathbf{a}$, the transpose of $\mathbf{a}$ is denoted as $\mathbf{a}^T$. For natural numbers $m$ and $n\geq m$, the set $\{y_m, y_{m+1}, \cdots, y_n\}$ is represented as $\{y_i\}_{i=m}^{n}$. For a matrix $G$, the entry of $G$ at the $i^{th}$ row and $j^{th}$ column is denoted as $G_{i,j}$.
We also express the nodes in a clustered DSS using a two-dimensional representation: in the structure illustrated in Fig. \ref{Fig:2dim_representation}, $N(l,j)$ represents the node at the $l^{th}$ row and the $j^{th}$ column. 
Finally, we recall definitions on the locally repairable codes (LRCs) in \cite{papailiopoulos2014locally, tamo2016optimal}.  As defined in \cite{tamo2016optimal}, an $(n,k,r)-$LRC represents a code of length  $n$, which is encoded from $k$ information symbols. Every coded symbol of the $(n,k,r)-$LRC can be regenerated by accessing at most $r$ other symbols. As defined in \cite{papailiopoulos2014locally}, an $(n,r,d,\mathcal{M},\alpha)-$LRC takes a file of size $\mathcal{M}$ and encodes it into $n$ coded symbols, where each symbol is composed of $\alpha$ bits. Moreover, any coded symbol can be regenerated by contacting at most $r$ other symbols, and the code has the minimum distance of $d$.


\begin{figure}[t]
	\centering
	\includegraphics[height=35mm]{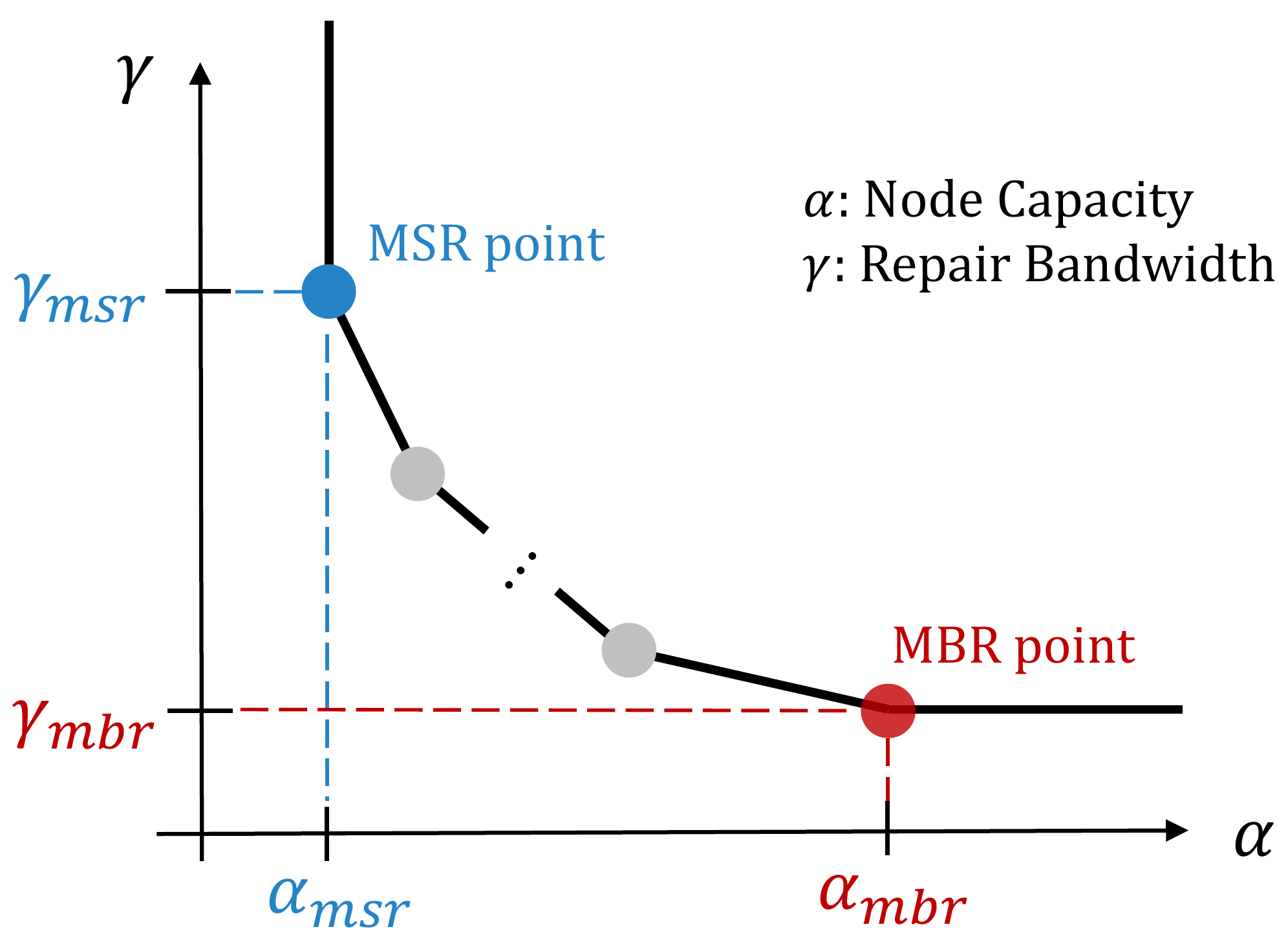}
	\caption{The optimal trade-off relationship between $\alpha$ and $\gamma$ in the clustered distributed storage modeled in \cite{sohn2017TIT}}
	\label{Fig:MBR_MSR_points}
\end{figure}

\begin{figure}[!t]
	\centering
	\includegraphics[height=22mm]{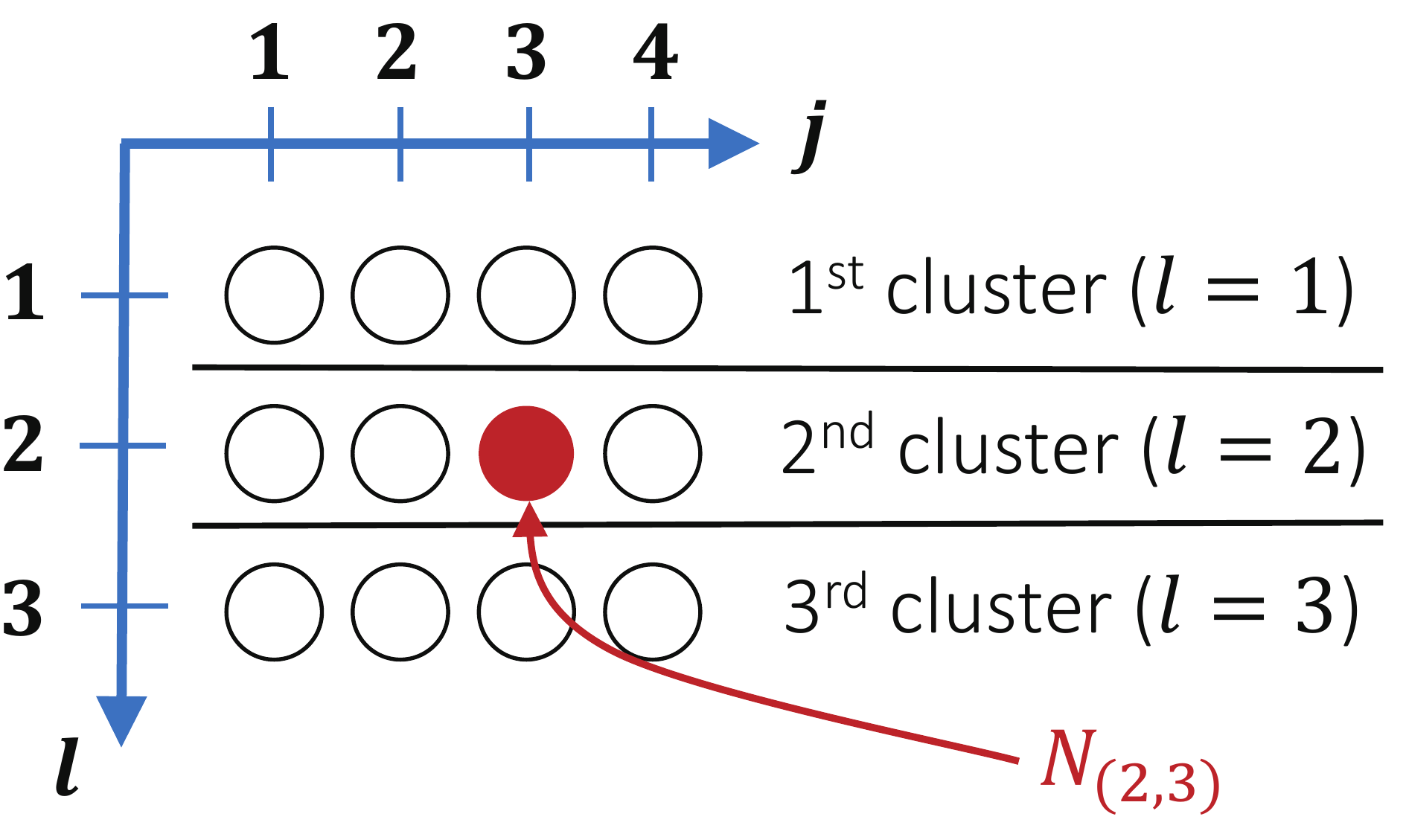}
	\caption{Two-dimensional representation of clustered distributed storage ($n=12, L=3, n_I = n/L = 4$)}
	\label{Fig:2dim_representation}
\end{figure}

\section{MSR Code Design for $\epsilon = 0$}

In this section, MSR codes for $\epsilon=0$ (\textit{i.e.}, $\beta_c = 0$) is designed. Under this setting, no cross-cluster communication is allowed in the node repair process.
First, the system parameters for the MSR point are examined. Second, two types of locally repairable codes (LRCs) suggested in \cite{papailiopoulos2014locally, tamo2016optimal} are proven to achieve the MSR point, under the settings of $n_I \divides k$ and $n_I \notdivides k$, respectively. 

\subsection{Parameter Setting for the MSR Point}

We consider the MSR point $(\alpha, \gamma)=(\alpha_\text{msr}, \gamma_\text{msr})$ which can reliably store file $\mathcal{M}$.
The following property specifies the system parameters for the $\epsilon=0$ case.

\begin{prop}\label{Prop:parameter_for_small_epsilon}
	Consider an [n,k,L] clustered DSS to reliably store file $\mathcal{M}$. The MSR point for $\epsilon=0$ is
	\begin{equation}\label{Eqn:parameters_for_small_epsilon}
	(\alpha_\text{msr}, \gamma_\text{msr}) = \left(\frac{\mathcal{M}}{k-q}, \frac{\mathcal{M}}{k-q} (n_I-1)\right),
	\end{equation}
	where $q$ is defined in (\ref{Eqn:quotient}).
	This point satisfies $\alpha=\beta_I$.
	\end{prop}
\begin{proof}
	See Appendix \ref{Section:proof_of_prop_param_small_epsilon}.
\end{proof}

\subsection{Code Construction for $n_I \divides k$}

We now examine how to construct an MSR code for the $n_I \divides k$ case. 
The following theorem shows that a locally repairable code constructed in \cite{papailiopoulos2014locally} with locality $r=n_I-1$ is a valid MSR code for $n_I \divides k$.

\begin{thm}[Exact-repair MSR Code Construction for $\epsilon=0, n_I \divides k$]\label{Thm:LRC1_achieves_MSR}
	Let $\mathds{C}$ be the $(n,r,d,\mathcal{M},\alpha)-$LRC explicitly constructed in \cite{papailiopoulos2014locally} for locality $r=n_I-1$. 
	Consider allocating coded symbols of $\mathds{C}$ in a $[n,k,L]-$clustered DSS, where $r+1=n_I$ nodes within the same repair group of $ \mathds{C}$ are located in the same cluster.	
	Then, the code $\mathds{C}$ is an MSR code for the $[n,k,L]-$ clustered DSS under the conditions of $\epsilon=0$ and $n_I \divides k$.
\end{thm}
\begin{proof}
	See Appendix \ref{Section:proof_of_LRC1_achieves_MSR}.
\end{proof}

Fig. \ref{Fig:MSR_for_epsilon0_divisible} illustrates an example of the MSR code for the $\epsilon=0$ and $n_I \divides k$ case, which is constructed using the LRC in \cite{papailiopoulos2014locally}. In the $[n,k,L]=[6,3,2]$ clustered DSS scenario, the parameters are set to
\begin{align*}
\alpha &= n_I = n/L = 3, \\
\mathcal{M} &= (k-q) \alpha = (k-\floor{k/n_I}) \alpha = 6.
\end{align*}
Thus, each storage node contains $\alpha=3$ symbols, while the $[n,k,L]$ clustered DSS aims to reliably store a file of size $\mathcal{M}=6$.
This code has two properties, 1) \textit{exact regeneration} and 2) \textit{data reconstruction}:
\begin{enumerate}
	\item Any failed node can be exactly regenerated by contacting $n_I-1 = 2$ nodes in the same cluster,
	\item Contacting any $k=3$ nodes can recover the original file $\{x_i^{(j)}: i \in [3], j \in [2] \}$ of size $\mathcal{M}=6$.
\end{enumerate}

The first property is obtained from the fact that $y_i^{(1)}, y_i^{(2)}$ and $s_i = y_i^{(1)}+ y_i^{(2)}$ form a $(3,2)$ MDS code for $i \in [6]$. The second property is obtained as follows. For contacting  arbitrary $k=3$ nodes, three distinct coded symbols $\{y_{i_1}^{(1)}, y_{i_2}^{(1)}, y_{i_3}^{(1)}\}$ having superscript one and three distinct coded symbols $\{y_{j_1}^{(2)}, y_{j_2}^{(2)}, y_{j_3}^{(2)}\}$ having superscript two can be obtained for some 
$i_1, i_2, i_3 \in [6]$ and 
$j_1, j_2, j_3 \in [6]$. 
From Fig. \ref{Fig:MDS_Precoding}, the information $\{y_{i_1}^{(1)}, y_{i_2}^{(1)}, y_{i_3}^{(1)}\}$ suffice to recover $x_1^{(1)}, x_2^{(1)}, x_3^{(1)}$. Similarly, the information $\{y_{j_1}^{(2)}, y_{j_2}^{(2)}, y_{j_3}^{(2)}\}$ suffice to recover $x_1^{(2)}, x_2^{(2)}, x_3^{(2)}$. This completes the proof for the second property.  
Note that this coding scheme is already suggested by the authors of \cite{papailiopoulos2014locally}, while the present paper proves that this code also achieves the MSR point of the $[n,k,L]$ clustered DSS, in the case of $\epsilon=0$ and $n_I \divides k$.

\begin{figure}
	\centering
	\subfloat[][MDS Precoding]{\includegraphics[width=80mm ]{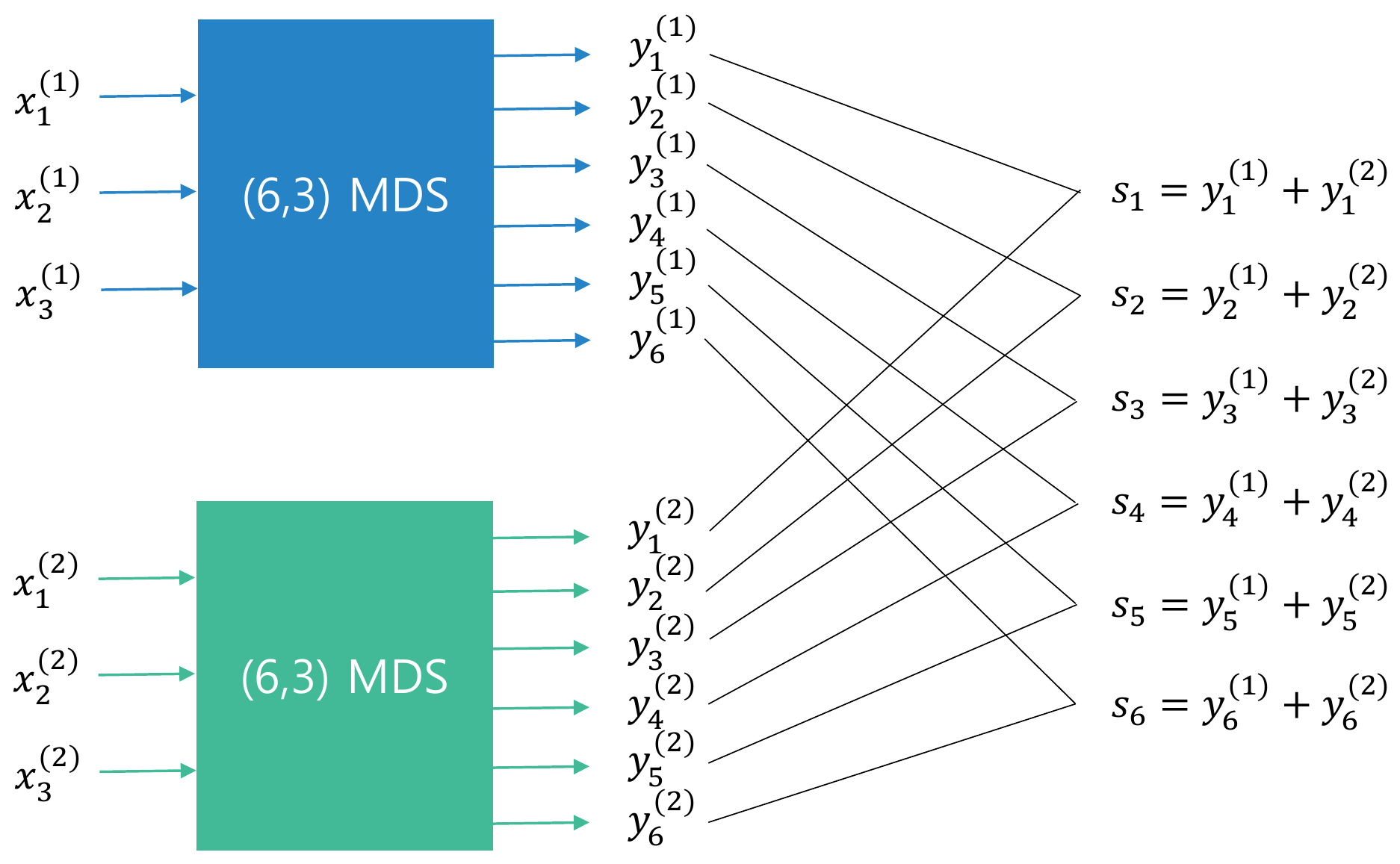}\label{Fig:MDS_Precoding}}
	\quad \quad
	\subfloat[][Allocation of coded symbols into $n$ nodes]{\includegraphics[width=90mm ]{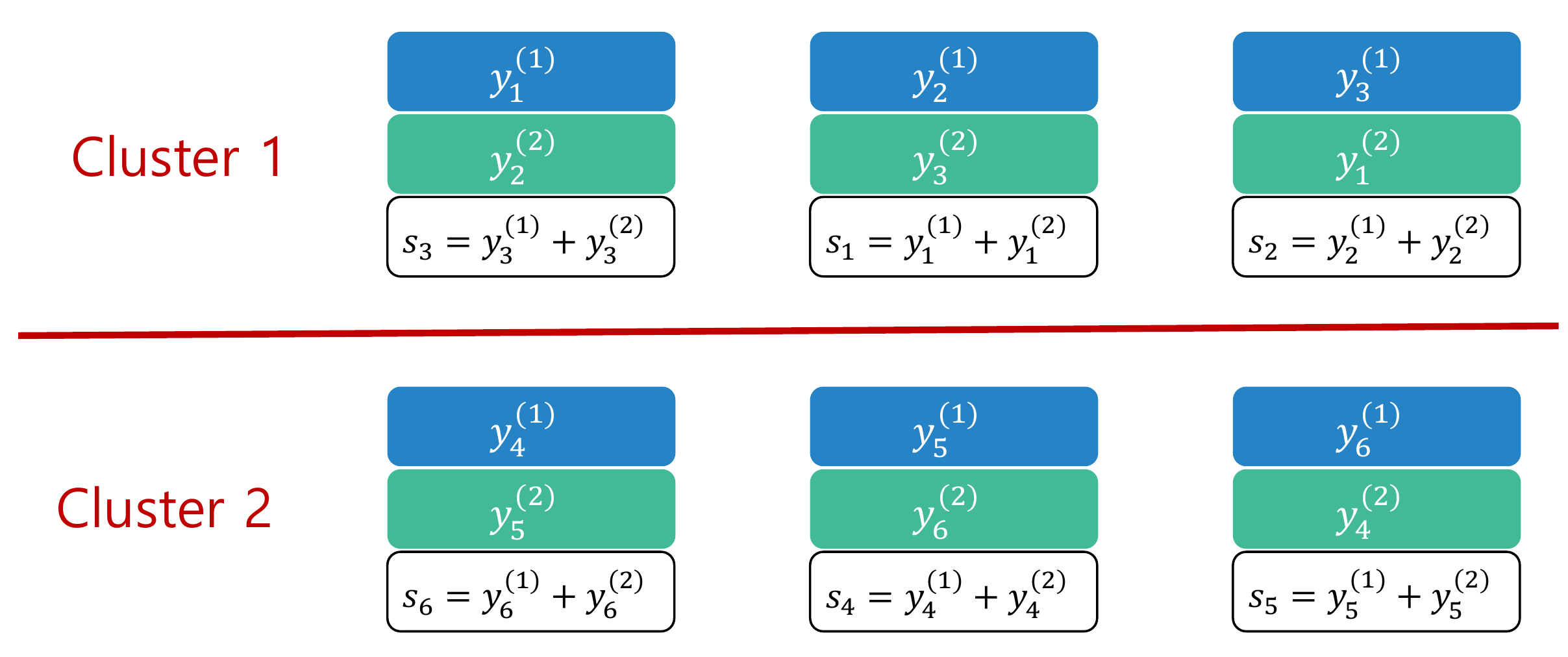}\label{Fig:Alloc_nodes}}
	\caption{MSR code for $\epsilon=0$ with $n_I \divides k$ ($n=6, k=3, L=2$). The construction rule follows the instruction in \cite{papailiopoulos2014locally}, while the concept of the \textit{repair group} in \cite{papailiopoulos2014locally} can be interpreted as the \textit{cluster} in the present paper.} 
	\label{Fig:MSR_for_epsilon0_divisible}
\end{figure}

\subsection{Code Construction for $n_I \notdivides k$}

Here we construct an MSR code when the given system parameters satisfy $n_I \notdivides k$. The theorem below shows that the optimal $(n,k-q,n_I-1)-$LRC designed in \cite{tamo2016optimal} is a valid MSR code when $n_I \notdivides k$ holds.

\begin{thm}[Exact-repair MSR Code Construction for $\epsilon=0, n_I \notdivides k$]\label{Thm:LRC2_achieves_MSR}
	Let $\mathds{C}$ be the $(n_0,k_0,r_0)-$LRC constructed in \cite{tamo2016optimal} for $n_0=n, k_0=k-q$ and $r_0=n_I-1$.
	Consider allocating the coded symbols of $\mathds{C}$ in a $[n,k,L]-$clustered DSS, where $r+1=n_I$ nodes within the same repair group of $ \mathds{C}$ are located in the same cluster.
	Then, $\mathds{C}$ is an MSR code for the $[n,k,L]-$clustered DSS under the conditions of $\epsilon=0$ and $n_I \notdivides k$.
\end{thm}
\begin{proof}
	See Appendix \ref{Section:proof_of_LRC2_achieves_MSR}.
\end{proof}

Fig. \ref{Fig:Epsilon0_nondivisible_construct} illustrates an example of code construction for the $n_I \notdivides k$ case. Without loss of generality, we consider $\alpha=1$ case; parallel application of this code multiple $\alpha$ times achieves the MSR point for general $\alpha \in \mathbb{N}$, where $\mathbb{N}$ is the set of positivie integers.
In the $[n=6,k=4,L=2]$ clustered DSS with $\epsilon=0$, the code and system parameters are:
\begin{align*}
[n_0,k_0,r_0]&=[n,k-q,n_I-1]=[6,3,2],\\
\alpha&=1, \\
\mathcal{M}&=(k-q)\alpha=(k-\floor{k/n_I})=3
\end{align*}
from Proposition \ref{Prop:parameter_for_small_epsilon}.
The code in Fig. \ref{Fig:Epsilon0_nondivisible_construct} satisfies the \textit{exact regeneration} and \textit{data reconstruction} properties:
\begin{enumerate}
	\item Any failed node can be exactly regenerated by contacting $n_I-1 = 2$ nodes in the same cluster,
	\item Contacting any $k=4$ nodes can recover the original file $\{x_i: i \in [3]\}$ of size $\mathcal{M}=3$.
\end{enumerate} 
Note that $\{y_i\}_{i=1}^3$ in Fig. \ref{Fig:Epsilon0_nondivisible_construct} is a set of coded symbols generated by a $(3,2)-$MDS code, and this statement also holds for $\{y_i\}_{i=4}^6$. This proves the first property. The second property is directly from the result of \cite{tamo2016optimal}, which states that the minimum distance of the $[n_0,k_0,r_0]-LRC$ is
\begin{align}
d&=n_0-k_0-\left\lceil\dfrac{k_0}{r_0}\right\rceil+2 =6-3-\ceil{3/2}+2=3. 
\end{align}
Note that the $[n_0,k_0,r_0]-LRC$ is already suggested by the authors of \cite{tamo2016optimal}, while the present paper proves that applying this code with $n_0=n, k_0=k-q, r_0=n_I-1$ achieves the MSR point of the $[n,k,L]-$clustered DSS, in the case of $\epsilon=0$ and $n_I \notdivides k$.

\begin{figure}
	\centering
	\subfloat[][Encoding structure]{\includegraphics[width=40mm ]{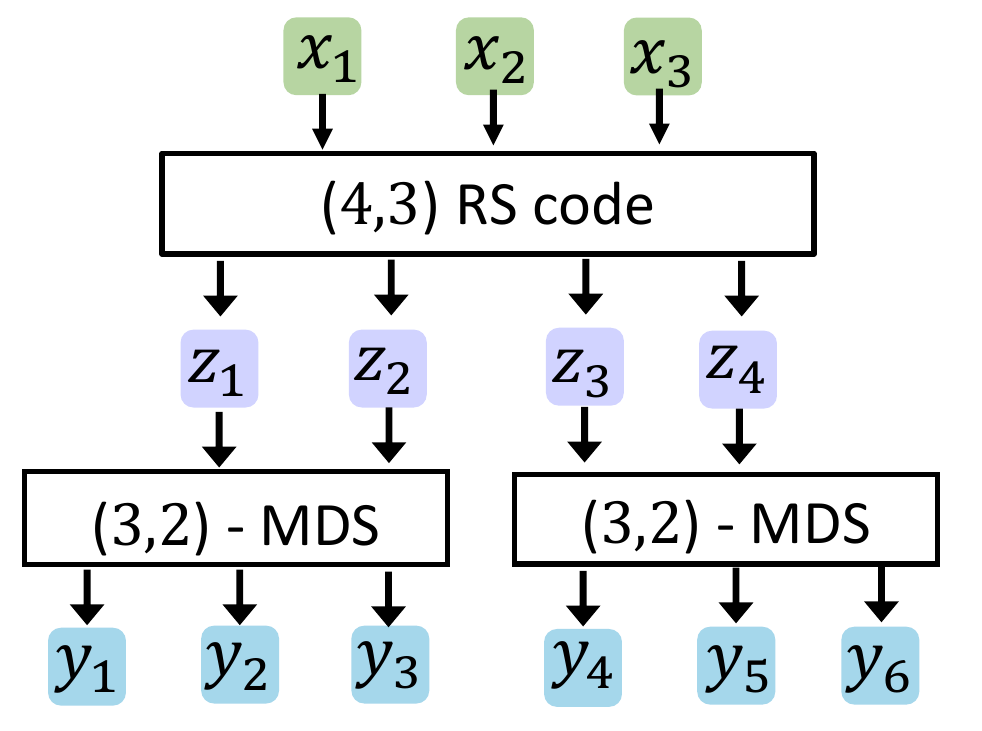}\label{Fig:Epsilon0_nondivisible_ENC}}
	\quad \quad
	\subfloat[][Allocation of coded symbols into $n$ nodes]{\includegraphics[width=40mm ]{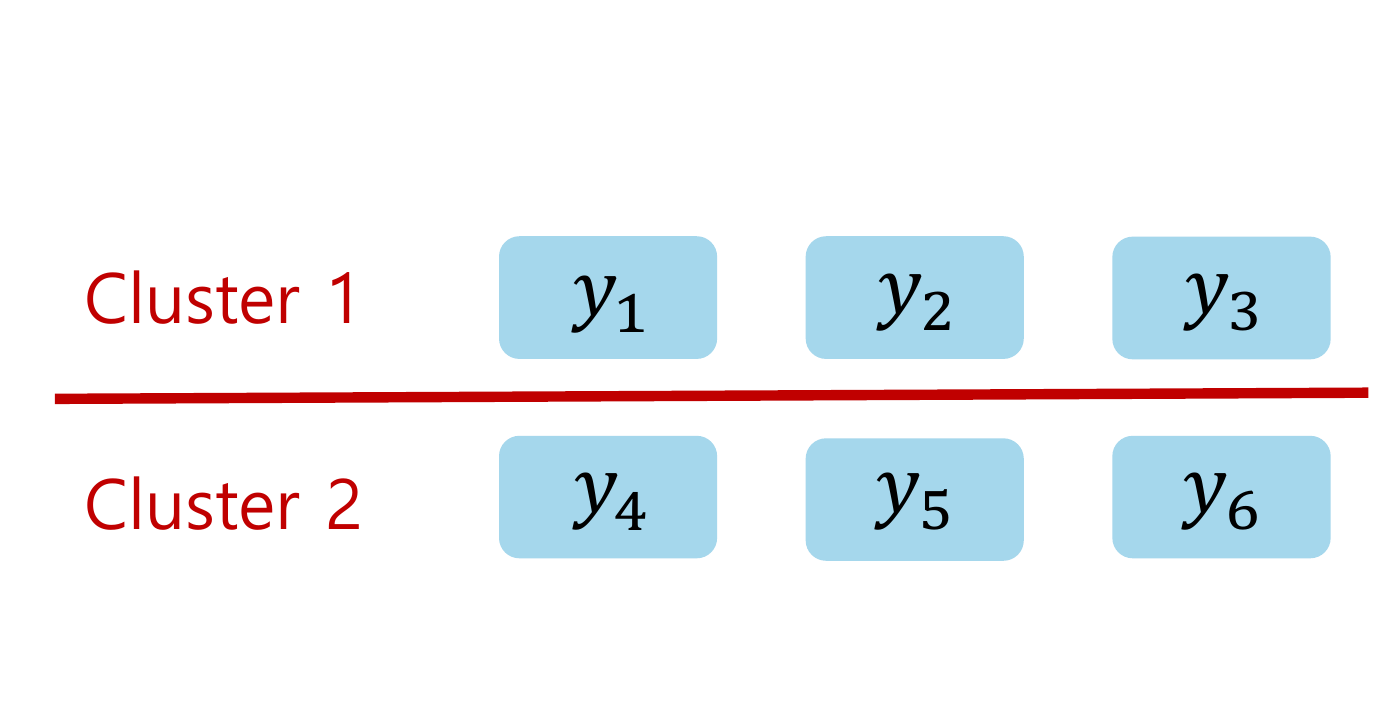}\label{Fig:Epsilon0_nondivisible_alloc}}
	\caption{MSR code for $\epsilon=0$ with $n_I \notdivides k$ case ($n=6, k=4, L=2$). The encoding structure follows from the instruction in \cite{tamo2016optimal}, which constructed $[n_0,k_0,r_0]-LRC$. This paper utilizes $[n,k-q,n_I-1]-LRC$ to construct MSR code for $[n,k,L]$ clustered DSS, in the case of $\epsilon=0$ with $n_I \notdivides k$.}
	\label{Fig:Epsilon0_nondivisible_construct}
\end{figure}

\section{MSR Code Design for $\epsilon = \frac{1}{n-k}$}

We propose an MSR code for $\epsilon=\frac{1}{n-k}$ in clustered DSSs. 
From (\ref{Eqn:alpha_MSR_large_epsilon_val}) and (\ref{Eqn:alpha_MSR_small_epsilon_val}), recall that $\frac{1}{n-k}$ is the minimum $\epsilon$ value which allows the minimum storage of $\alpha_\text{msr}=\mathcal{M}/k$.
First, we obtain the system parameters for the MSR point. Second, we design a coding scheme which is shown to be an MSR code under the conditions of $n=2k$ and $L=2$.

\subsection{Parameter Setting for the MSR Point}

The following property specifies the system parameters for the $\epsilon=1/(n-k)$ case. Without a loss of generality, we set the cross-cluster repair bandwidth as $\beta_c=1$. 

\begin{prop}\label{Prop:parameter_for_large_epsilon}
	The MSR point for $\epsilon=1/(n-k)$ is
	\begin{equation}\label{Eqn:parameters_for_large_epsilon}
	(\alpha_\text{msr}, \gamma_\text{msr}) = \left(\frac{\mathcal{M}}{k}, \frac{\mathcal{M}}{k} \left(n_I - 1 + \frac{n-n_I}{n-k}\right)\right).
	\end{equation}
	This point satisfies
	$\alpha=\beta_I = n-k$ and $\mathcal{M} = k(n-k)$. 
\end{prop}
\begin{proof}
	See Appendix \ref{Section:proof_of_prop_param_large_epsilon}.
\end{proof}

\subsection{Code Construction for $[n,k,L]=[2k,k,2]$}

Here, we construct an MSR code under the constraints of $n=2k$ and $L=2$. Since we consider the $n=2k$ case, the system parameters in Proposition \ref{Prop:parameter_for_large_epsilon} are set to
\begin{align}
\alpha &= \beta_I = n-k=k, \label{Eqn:alpha_MSR_large_epsilon}\\
\mathcal{M} &= k\alpha = k^2.  \nonumber\label{Eqn:filesize_MSR_large_epsilon}
\end{align}

\begin{construction}\label{Construct:MSR_for_2k_k_2}
	Suppose that we are given $\mathcal{M}=k^2$ source symbols $\{m_{i,j}: i,j \in [k]\}$. Moreover, let the encoding matrix
	\begin{equation}\label{Eqn:ENC_MAT}
	G=
	\begin{bmatrix}
	G_1^{(1)} & G_1^{(2)} & \cdots & G_1^{(k)}\\
	G_2^{(1)} & G_2^{(2)} & \cdots & G_2^{(k)}\\
	\vdots & \vdots & \ddots & \vdots \\
	G_k^{(1)} & G_k^{(2)} & \cdots & G_k^{(k)}
	\end{bmatrix}
	\end{equation}
	be a $k^2 \times k^2$ matrix, where each encoding sub-matrix $G_i^{(j)}$ is a $k \times k$ matrix. 	
	 For $j \in [k]$, node $N(1,j)$ stores $\mathbf{m}_j$ and node $N(2,j)$ stores $\mathbf{p}_j$, where 
	\begin{align}
	\mathbf{m}_{i}&=[m_{i,1}, \cdots, m_{i,k}]^T, \\
	\mathbf{p}_{i}&=[p_{i,1}, \cdots, p_{i,k}]^T=\sum_{j=1}^k \mathbf{m}_{j}^T G_i^{(j)}.\label{Eqn:parity}
	\end{align}
\end{construction}

\begin{remark}
	The code generated in Construction \ref{Construct:MSR_for_2k_k_2} satisfies the followings:
	\begin{enumerate}[label=(\alph*)]
		\item Every node in cluster $1$ contains $k$ message symbols.
		\item Every node in cluster $2$ contains $k$ parity symbols.
	\end{enumerate}
\end{remark}

Note that this remark is consistent with  (\ref{Eqn:alpha_MSR_large_epsilon}), which states $\alpha=k$. Under this construction, we have the following theorem, which specifies the MSR construction rule for the $[n=2k,k,L=2]-$DSS with $\epsilon=1/(n-k)$.

\begin{thm}[Exact-repair MSR Code Construction for $\epsilon=\frac{1}{n-k}$]\label{Thm:MSR_code_for_2k_k_2_1/k}
	If all square sub-matrices of $G$ are invertible, the code designed by Construction \ref{Construct:MSR_for_2k_k_2} is an MSR code for $[n,k,L]=[2k,k,2]-$DSS with $\epsilon=1/(n-k)$.
\end{thm}
\begin{proof}
	See Appendix \ref{Section:proof_MSR_code_for_2k_k_2_1/k}.
\end{proof}

The following result suggests an explicit construction of an MSR code using the finite field.
\begin{corollary}
	Applying Construction \ref{Construct:MSR_for_2k_k_2} with encoding matrix $G$ set to the $k^2 \times k^2$ Cauchy matrix \cite{10.2307/j.ctt7t833} achieves the MSR point for an $[n=2k,k,L=2]-$DSS. A finite field of size $2k^2$ suffices to design $G$.
\end{corollary}
\begin{proof}
	The proof is directly from Theorem \ref{Thm:MSR_code_for_2k_k_2_1/k} and the fact that all sub-matrices of a Cauchy matrix has full rank, as stated in \cite{shah2010explicit}. Moreover, the Cauchy matrix of size $n \times n$ can be constructed using a finite field of size $2n$, according to \cite{suh2011exact}.
\end{proof}

An example of MSR code designed by Construction \ref{Construct:MSR_for_2k_k_2} is illustrated in Fig. \ref{Fig:MSR_large_epsilon_toy}, in the case of $n=4, k=2, L=2$. 
This coding scheme utilizes a Cauchy matrix 
\begin{equation}
G=
\begin{bmatrix}
7 & 2 & 3 & 4 \\
2 & 7 & 4 & 3\\
3 & 4 & 7 & 2\\
4 & 3 & 2 & 7
\end{bmatrix}
\end{equation}
using the finite field $GF(2^3)$ with the primitive polynomial $x^3+x+1$.
The element $a\alpha^2+b\alpha+c$ in $GF(2^3)$ is denoted by the decimal number of $(abc)_2$, where $\alpha$ is the primitive element. For example, $\alpha+1$ is denoted by $3=(011)_2$ in the generator matrix $G$.
When $[n,k,L,\epsilon]=[4,2,2,1/2]$, the system parameters are
\begin{equation*}
\alpha=2, \mathcal{M}=4, \beta_I = 2, \beta_c = 1
\end{equation*}
from Proposition \ref{Prop:parameter_for_large_epsilon}, which holds for the example in Fig. \ref{Fig:MSR_large_epsilon_toy}.
Here we show that the proposed coding scheme satisfies two properties: 1) exact regeneration of any failed node and 2) recovery of $\mathcal{M}=4$ message symbols $\{m_{1,1}, m_{1,2}, m_{2,1}, m_{2,2}\}$ by contacting any $k=2$ nodes.

\begin{figure}[!t]
	\centering
	\includegraphics[width=80mm]{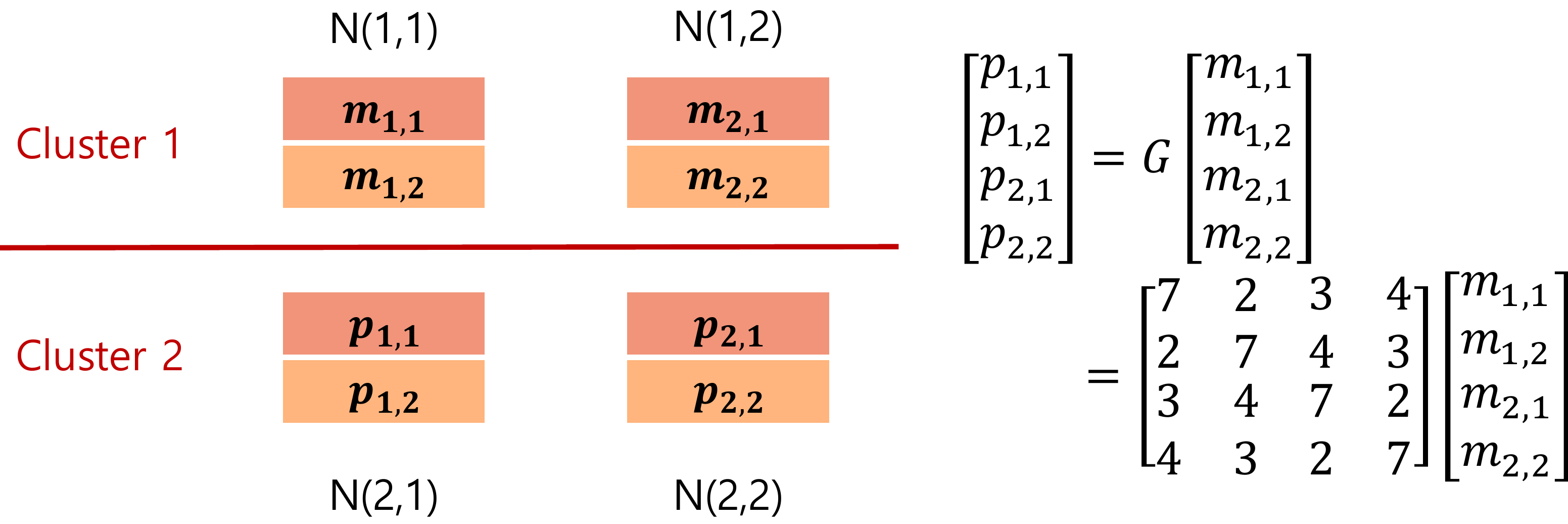}
	\caption{MSR example for $n=4,k=2,L=2$}
	\label{Fig:MSR_large_epsilon_toy}
\end{figure}

1) \textit{Exact regeneration}: Fig. \ref{Fig:MSR_large_epsilon_toy_repair} illustrates the regeneration process. Suppose that node $N(1,1)$ containing the message $\mathbf{m}_1=[m_{1,1}, m_{1,2}]$ fails. Then, node $N(1,2)$ transmits $\beta_I = 2$ symbols, $m_{2,1}$ and $m_{2,2}$. Nodes $N(2,1)$ and $N(2,2)$ transmit $\beta_c = 1$ symbol each, for example $p_{1,1}$ and $p_{2,2}$, respectively. Then, from the received symbols of $m_{2,1}, m_{2,2}, p_{1,1}, p_{2,2}$ and matrix $G$, we obtain
\begin{equation*}
\begin{bmatrix}
y_1 \\
y_2
\end{bmatrix}
\coloneqq
\begin{bmatrix}
p_{1,1}-G_{1,3}m_{2,1}-G_{1,4}m_{2,2} \\
p_{2,2}-G_{4,3}m_{2,1}-G_{4,4}m_{2,2} 
\end{bmatrix}
=
\begin{bmatrix}
7 & 2\\
4 & 3
\end{bmatrix}
\begin{bmatrix}
m_{1,1}\\
m_{1,2}
\end{bmatrix}.
\end{equation*}
Thus, the contents of the failed node can be regenerated by
\begin{equation*}
\begin{bmatrix}
m_{1,1}\\
m_{1,2}
\end{bmatrix}
=\begin{bmatrix}
7 & 2\\
4 & 3
\end{bmatrix}^{-1}
\begin{bmatrix}
y_1\\
y_2
\end{bmatrix}
=\begin{bmatrix}
3 & 2\\
4 & 7
\end{bmatrix}
\begin{bmatrix}
y_1\\
y_2
\end{bmatrix}
\end{equation*}
where the matrix inversion is over $GF(2^3)$.
Note that the exact regeneration property holds irrespective of the contents transmitted by $N(2,1)$ and $N(2,2)$, since the encoding matrix is a Cauchy matrix, all submatrices of which are invertible. 

2) \textit{Data recovery}: First, if DC contacts two systematic nodes, the proof is trivial. Second, contacting two parity nodes can recover the original message since $G$ is invertible. Third, suppose that DC contacts one systematic node and one parity node, for example, $N(1,1)$ and $N(1,4)$. Then, DC can retrieve message symbols $m_{1,1}, m_{1,2}$ and parity symbols $p_{2,1}, p_{2,2}$. Using the retrieved symbols and the information on the encoding matrix $G$, DC additionally obtains
\begin{equation*}
\begin{bmatrix}
z_1 \\
z_2
\end{bmatrix}
\coloneqq
\begin{bmatrix}
p_{2,1}-G_{3,1}m_{1,1}-G_{3,2}m_{1,2} \\
p_{2,2}-G_{4,1}m_{1,1}-G_{4,2}m_{1,2} 
\end{bmatrix}
=
\begin{bmatrix}
7 & 2\\
2 & 7
\end{bmatrix}
\begin{bmatrix}
m_{2,1}\\
m_{2,2}
\end{bmatrix}.
\end{equation*}
Thus, DC obtains 
\begin{equation*}
\begin{bmatrix}
m_{2,1}\\
m_{2,2}
\end{bmatrix}
=\begin{bmatrix}
7 & 2\\
2 & 7
\end{bmatrix}^{-1}
\begin{bmatrix}
z_1\\
z_2
\end{bmatrix}
=\begin{bmatrix}
1 & 3\\
3 & 1
\end{bmatrix}
\begin{bmatrix}
z_1\\
z_2
\end{bmatrix},
\end{equation*}
which completes the data recovery property of the suggested code.

\begin{figure}[!t]
	\centering
	\includegraphics[width=85mm]{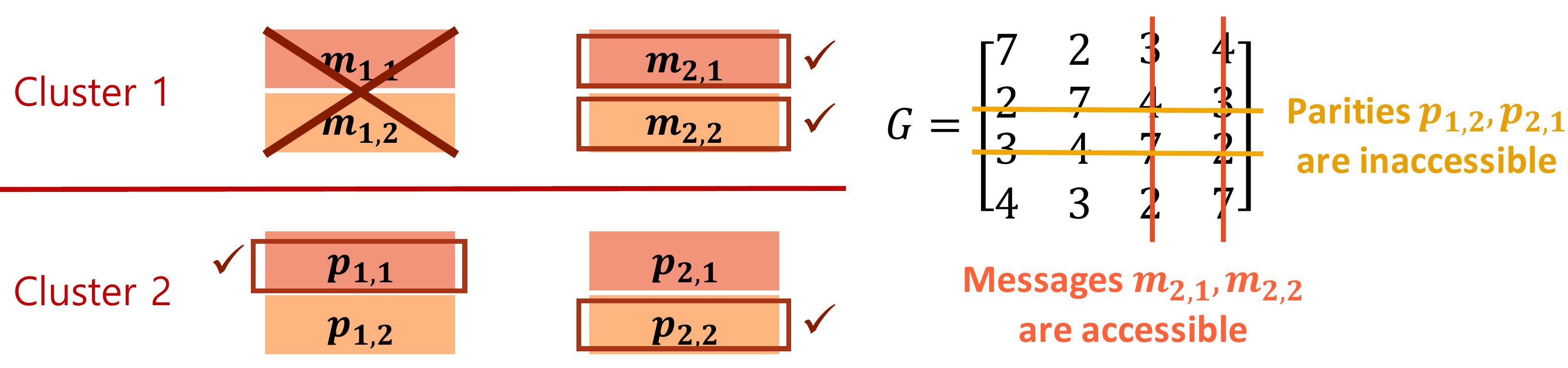}
	\caption{Repairing a failed node in proposed MSR code example for $n=4,k=2,L=2$}
	\label{Fig:MSR_large_epsilon_toy_repair}
\end{figure}

\section{Conclusion}
A class of MSR codes for clustered distributed storage modeled in \cite{sohn2016capacity} has been constructed. The proposed coding schemes can be applied in practical data centers with multiple racks, where the available cross-rack bandwidth is limited compared to the intra-rack bandwidth. 
Two important cases of $\epsilon=0$ and $\epsilon=1/(n-k)$ are considered, where $\epsilon=\beta_c/\beta_I$ represents the ratio of available cross-
to intra-cluster repair bandwidth. 
Under the constraint of zero cross-cluster repair bandwidth ($\epsilon=0$), appropriate application of two locally repairable codes suggested in \cite{papailiopoulos2014locally, tamo2016optimal} is shown to achieve the MSR point of clustered distributed storage. Moreover, an explicit MSR coding scheme is suggested for $\epsilon=1/(n-k)$, when the system parameters satisfy $n=2k$ and $L=2$. The proposed coding scheme can be implemented in a finite field, by using a Cauchy generator matrix.



\appendices

\numberwithin{equation}{section}

\section{Proof of Theorem \ref{Thm:LRC1_achieves_MSR}}\label{Section:proof_of_LRC1_achieves_MSR}

We focus on code $\mathds{C}$, the explicit ($n,r,d,\mathcal{M},\alpha$)-LRC constructed in Section V of \cite{papailiopoulos2014locally}. This code has the parameters
\begin{equation}\label{Eqn:LRC1_param}
(n,r,d=n-k+1,\mathcal{M},\alpha=\frac{r+1}{r}\frac{\mathcal{M}}{k}),
\end{equation}
where $r$ is the repair locality and $d$ is the minimum distance, and other parameters ($n,\mathcal{M},\alpha$) have physical meanings identical to those in the present paper.
By setting $r=n_I-1$, the code has node capacity of
\begin{equation}\label{Eqn:alpha_MSR_small_epsilon}
\alpha = \frac{n_I}{n_I-1}\frac{\mathcal{M}}{k} = \frac{\mathcal{M}}{k(1-1/n_I)}=\frac{\mathcal{M}}{k-q}
\end{equation}
where the last equality holds from the $n_I \divides k$ condition and the definition of $q$ in (\ref{Eqn:quotient}).

We first prove that any node failure can be exactly regenerated by using the system parameters in \eqref{Eqn:parameters_for_small_epsilon}.
According to the description in Section V-B of \cite{papailiopoulos2014locally}, \textit{any} node is contained in a unique corresponding repair group of size $r+1=n_I$, so that a failed node can be exactly repaired by contacting $r=n_I-1$ other nodes in the same repair group. This implies that a failed node does not need to contact other repair groups in the exact regeneration process. 
By setting each repair group as a cluster (note that each cluster contains $n_I=n/L$ nodes), we can achieve 
\begin{equation}\label{Eqn:betac_MSR_small_epsilon}
\beta_c = 0.
\end{equation}
Moreover, Section V-B of \cite{papailiopoulos2014locally} illustrates that the exact regeneration of a failed node is possible by contacting the \textit{entire} symbols contained in $r=n_I-1$ nodes in the same repair group, and applying the XOR operation. This implies $\beta_I = \alpha$, which result in
\begin{equation}\label{Eqn:betai_MSR_small_epsilon}
\gamma = (n_I-1)\beta_I = (n_I-1) \frac{\mathcal{M}}{k-q},
\end{equation}
combined with \eqref{Eqn:gamma} and \eqref{Eqn:alpha_MSR_small_epsilon}.
From (\ref{Eqn:alpha_MSR_small_epsilon}) and (\ref{Eqn:betai_MSR_small_epsilon}), we can conclude that code $\mathds{C}$ satisfies the exact regeneration of any failed node using the parameters in \eqref{Eqn:parameters_for_small_epsilon}.

Now we prove that contacting any $k$ nodes suffices to recover original data in the clustered DSS with code $\mathds{C}$ applied. Note that the minimum distance is $d=n-k+1$ from (\ref{Eqn:LRC1_param}). Thus, the information from $k$ nodes suffices to pick the correct codeword. This completes the proof of Theorem \ref{Thm:LRC1_achieves_MSR}.

\section{Proof of Theorem \ref{Thm:LRC2_achieves_MSR}}\label{Section:proof_of_LRC2_achieves_MSR}

We first prove that the code $\mathds{C}$ has minimum distance of $d=n-k+1$, which implies that the original file of size $\mathcal{M}=k-q$ can be recovered by contacting arbitrary $k$ nodes. Second, we prove that any failed node can be exactly regenerated under the setting of (\ref{Eqn:parameters_for_small_epsilon}).
Recall that the $[n_0,k_0,r_0]-$LRC constructed in \cite{tamo2016optimal} has the following property, as stated in Theorem 1 of \cite{tamo2016optimal}:

\begin{lemma}[Theorem 1 of \cite{tamo2016optimal}]\label{Lemma:Result_of_Tamo}
	The code constructed in \cite{tamo2016optimal} has locality $r_0$ and optimal minimum distance $d=n_0-k_0-\ceil{\frac{k_0}{r}}+2$, when $(r_0+1)\divides n_0$.
\end{lemma}
Note that we consider code $\mathds{C}$ of optimal $[n_0,k_0,r_0]=[n,k-q,n_I-1]-$LRC. Since $r_0+1=n_I$ divides $n_0=n$, Lemma \ref{Lemma:Result_of_Tamo} can be applied. 
The result of Lemma \ref{Lemma:Result_of_Tamo} implies that the minimum distance of $\mathds{C}$ is
\begin{align}\label{Eqn:dmin_tamo}
d &= n-(k-q)-\left\lceil\dfrac{k-q}{n_I-1}\right\rceil+2.
\end{align}
Since we consider the $n_I \notdivides k$ case, we have
\begin{equation}\label{Eqn:k_case2}
k = qn_I + m, \quad \quad (0 < m \leq n_I-1)
\end{equation}
from (\ref{Eqn:remainder}). 
Inserting (\ref{Eqn:k_case2}) into (\ref{Eqn:dmin_tamo}), we have
\begin{align}\label{dmin_case2}
d &= n-(k-q)-\left\lceil\dfrac{(n_I-1)q+m}{n_I-1}\right\rceil+2 \nonumber\\
&= n-(k-q)-(q+1)+2 = n-k+1,
\end{align}
where the second last equality holds since $0 < m \leq n_I-1$ from (\ref{Eqn:k_case2}). Thus, this proves that contacting arbitrary $k$ nodes suffices to recover the original source file.

\begin{figure}[!t]
	\centering
	\includegraphics[width=90mm]{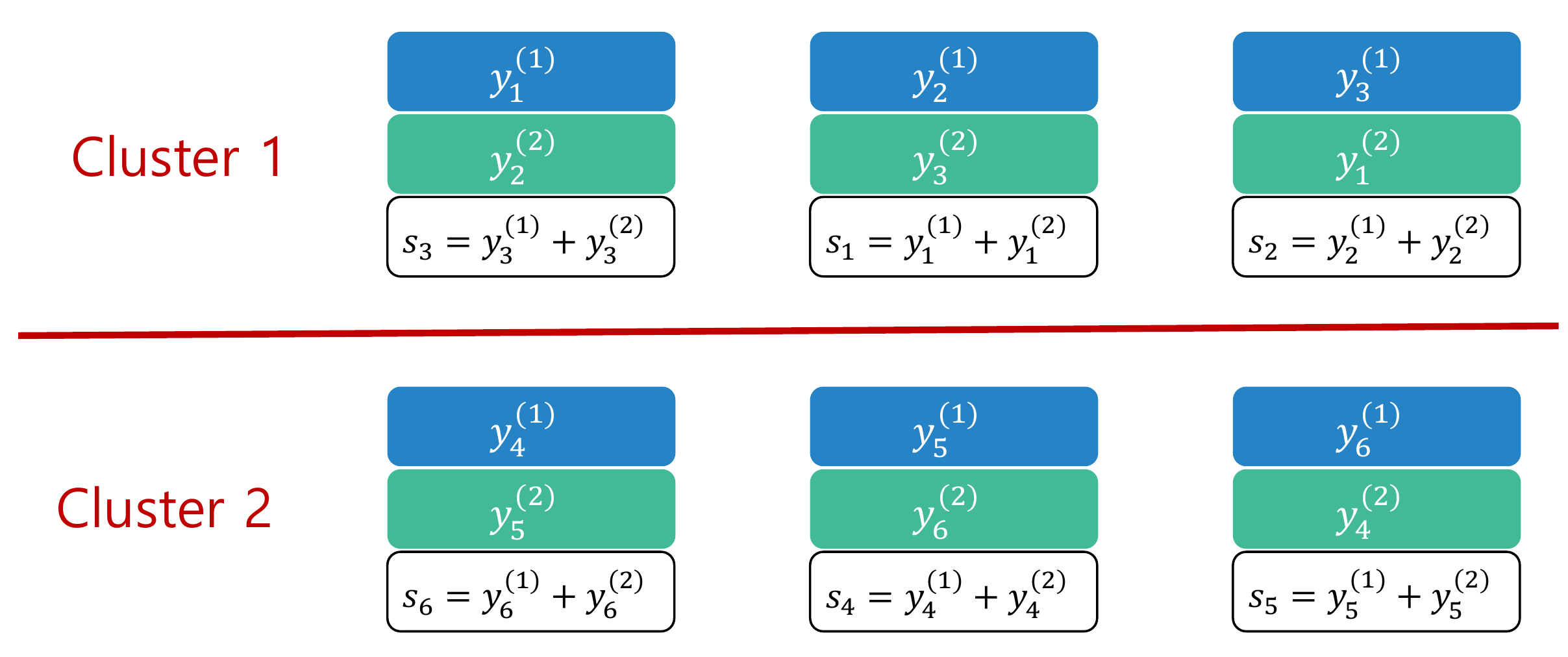}
	\caption{Code construction for $\epsilon=0, n_I \notdivides k$ case}
	\label{Fig:Epsilon_nondivisible}
\end{figure}

Now, all we need to prove is that any failed node can be exactly regenerated under the setting of system parameters specified in Proposition \ref{Prop:parameter_for_small_epsilon}.
According to the rule illustrated in \cite{tamo2016optimal}, the construction of code $\mathds{C}$ can be shown as in Fig. \ref{Fig:Epsilon_nondivisible}.
First, we have $\mathcal{M}=k-q$ source symbols $\{x_i\}_{i=1}^{k-q}$ to store reliably. By applying a $(T,k-q)$ Reed-Solomon code to the source symbols, we obtain $\{z_i\}_{i=1}^{t}$ where $T \coloneqq L(n_I-1)$. Then, we partition $\{z_i\}_{i=1}^T$ symbols into $L$ groups, where each group contains $(n_I-1)$ symbols. Next, each group of $\{z_i\}$ symbols is encoded by an $(n_I, n_I-1)-$MDS code, which result in a group of $n_I$ symbols of $\{y_i\}$. Finally, we store symbol $y_{n_I(l-1)+j}$ in node $N(l,j)$. By this allocation rule, $y_i$ symbols in the same group are located in the same cluster.

Assume that $N(l,j)$, the $j^\text{th}$ node at $l^\text{th}$ cluster,  containing $y_{n_I(l-1)+j}$ symbol fails
for $l \in [L]$ and $j \in [n_I]$.
From Fig. \ref{Fig:Epsilon_nondivisible}, we know that $(n_I-1)$ symbols of $\{y_{n_I(l-1)+s}\}_{s=1, s\neq j}^{n_I}$ stored in $l^{th}$ cluster can decode the $(n_I, n_I-1)-$MDS code for group $l$. Thus, the contents of $y_{n_I(l-1)+j}$ can be recovered by retrieving symbols from nodes in the the $l^{th}$ cluster (i.e., the same cluster where the failed node is in). This proves the ability of exactly regenerating an arbitrary failed node. 
The regeneration process satisfies
\begin{equation}\label{Eqn:Epsilon0_Regeneration}
\beta_c = 0, \beta_I = \alpha.
\end{equation}
Moreover, note that the code in Fig. \ref{Fig:Epsilon_nondivisible} has
\begin{equation}\label{Eqn:Epsilon0_capacity}
\mathcal{M} = (k-q)\alpha
\end{equation}
source symbols. 
Since parameters obtained in (\ref{Eqn:Epsilon0_Regeneration}) and (\ref{Eqn:Epsilon0_capacity}) are consistent with Proposition \ref{Prop:parameter_for_small_epsilon}, we can confirm that code $\mathds{C}$ is a valid MSR point under the conditions $\epsilon = 0$ and $n_I \notdivides k$.

\section{Proof of Theorem \ref{Thm:MSR_code_for_2k_k_2_1/k}}\label{Section:proof_MSR_code_for_2k_k_2_1/k}

\begin{figure}[!t]
	\centering
	\includegraphics[width=90mm]{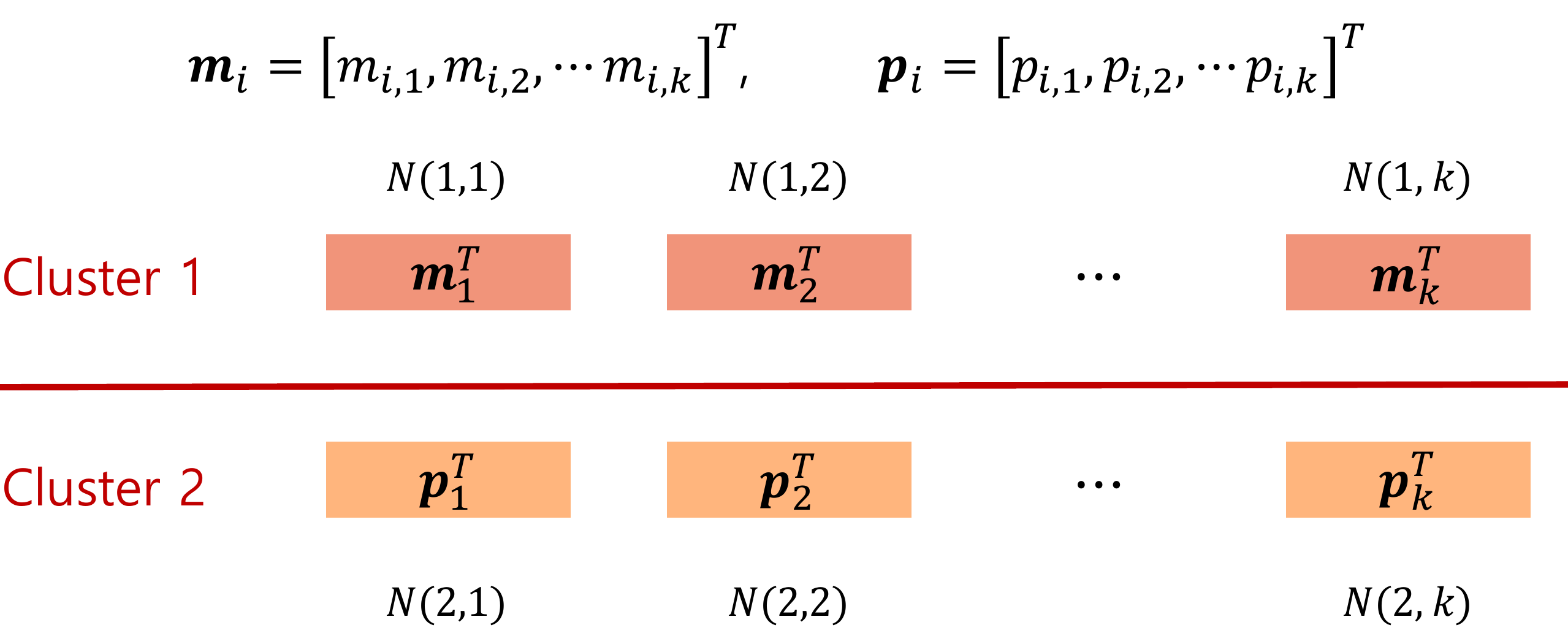}
	\caption{Code construction for $[n,k,L]=[2k,k,2]-$clustered DSS when $\epsilon=1/(n-k)$ }
	\label{Fig:proof_thm3}
\end{figure}

Recall that the code designed by Construction $\ref{Construct:MSR_for_2k_k_2}$ allocates systematic nodes at $1^{st}$ cluster and parity nodes at $2^{nd}$ cluster, as illustrated in Fig. \ref{Fig:proof_thm3}.
Moreover, recall that the system parameters for $[n,k,L]=[2k,k,2]-$DSS with $\epsilon=1/(n-k)$ are
\begin{align}
\alpha &=\beta_I = k, \quad 
\beta_c = 1,
\end{align}
from Proposition \ref{Prop:parameter_for_large_epsilon} and the definition of $\epsilon=\beta_c/\beta_I$.
First, we show that exact regeneration of systematic nodes (in the first cluster) is possible using $\beta_I = k, \beta_c = 1$ in the $[n,k,L]=[2k,k,2]$ DSS with Construction \ref{Construct:MSR_for_2k_k_2}.
We use the concept of the \textit{projection vector} to illustrate the repair process. 
For $l \in [L]$, let $\mathbf{v}_{i,j}^{(l)}$ be the $l^{th}$ projection vector assigned for 
$N(1,j)$, in repairing $N(1,i)$. Similarly, let  $\mathbf{v}_{i,j}$ be the projection vector assigned for 
$N(2,j)$, in repairing $N(1,i)$. 
Assume that the node $N(1,i)$ containing $\mathbf{m}_i = [m_{i,1}, m_{i,2}, \cdots, m_{i,k}]^T $ fails. Then, node $N(1,j)$ transmits $\beta_I = k$ symbols $\{\mathbf{m}_j^T \mathbf{v}_{i,j}^{(l)}\}_{l=1}^k$, while node $N(2,j)$ transmits $\beta_c = 1$ symbol $\mathbf{p}_j^T \mathbf{v}_{i,j}$. For simplicity, we set 
$\mathbf{v}_{i,j}^{(l)}=\mathbf{e}_l$ and $\mathbf{v}_{i,j}=\mathbf{e}_k$, where $\mathbf{e}_i$ is the $k$-dimensional standard basis vector with a $1$ in the $i^{th}$ coordinate and $0's$ elsewhere.
This means that node $N(1,j)$ transmits $k$ symbols $\mathbf{m}_j=[m_{j,1},m_{j,2},\cdots,m_{j,k}]^T$ it contains, while $N(2,j)$ transmits the last symbol it contains, i.e., the symbol $p_{j,k}$.
Thus, the newcomer node for regenerating systematic node $N(1,i)$ obtains the following information
\begin{equation}
M_i \coloneqq \{m_{j,s}: j \in [k] \setminus \{i\}, s \in [k] \} \cup \{p_{j,k}\}_{j=1}^k.
\end{equation}
We now show how the newcomer node regenerates $\mathbf{m}_i = [m_{i,1}, m_{i,2}, \cdots, m_{i,k}]^T $ using information $M_i$.
Recall that the parity symbols and message symbols are related as in the following $k^2$ equations:
\begin{equation}\label{Eqn:parity_message_relation}
\begin{bmatrix}
\mathbf{p}_1 \\
\mathbf{p}_2\\
\vdots \\
\mathbf{p}_k
\end{bmatrix}
= G 
\begin{bmatrix}
\mathbf{m}_1 \\
\mathbf{m}_2\\
\vdots \\
\mathbf{m}_k
\end{bmatrix}
\end{equation}
obtained from (\ref{Eqn:ENC_MAT}) and (\ref{Eqn:parity}).
Among these $k^2$ parity symbols, $k$ parity symbols received by the newcomer node can be expressed as
\begin{equation}\label{Eqn:parity_k}
\begin{bmatrix}
p_{1,k}\\
p_{2,k}\\
\vdots\\
p_{k,k}
\end{bmatrix}
=
\begin{bmatrix}
G_{k,1} & G_{k,2} & \cdots & G_{k,k^2}\\
G_{2k,1} & G_{2k,2} & \cdots & G_{2k,k^2}\\
\vdots & \vdots & \ddots & \vdots \\
G_{k^2,1} & G_{k^2,2} & \cdots & G_{k^2,k^2}
\end{bmatrix}
\begin{bmatrix}
m_{1,1}\\
m_{1,2}\\
\vdots\\
m_{k,k}
\end{bmatrix},
\end{equation}
where the matrix in \eqref{Eqn:parity_k} is generated by removing $k(k-1)$ rows from $G$.
Since we are aware of $k(k-1)$ message symbols of $\{m_{j,s}: j \in [k] \setminus \{i\}, s \in [k] \}$ and the entries of $G$ matrix, subtracting the constant known values from (\ref{Eqn:parity_k}) results in
\begin{equation}\label{Eqn:parity_k_reduced}
\begin{bmatrix}
y_1\\
y_2\\
\vdots\\
y_k
\end{bmatrix}
=
\begin{bmatrix}
G_{k,(i-1)k+1} & G_{k,(i-1)k+2} & \cdots & G_{k,ik}\\
G_{2k,(i-1)k+1} & G_{2k,(i-1)k+2} & \cdots & G_{2k,ik}\\
\vdots & \vdots & \ddots & \vdots \\
G_{k^2,(i-1)k+1} & G_{k^2,(i-1)k+2} & \cdots & G_{k^2,ik}
\end{bmatrix}
\begin{bmatrix}
m_{i,1}\\
m_{i,2}\\
\vdots\\
m_{i,k}
\end{bmatrix}
\end{equation}
where
\begin{equation}
y_l \coloneqq p_{l,k} - \sum_{j=1 \\ j\neq i}^{k}\sum_{s=1}^{k} G_{lk,(j-1)k+s}m_{j,s}
\end{equation}
for $l \in [k]$.
Note that the matrix in \eqref{Eqn:parity_k_reduced} can be obtained by removing $k(k-1)$ columns from the matrix in \eqref{Eqn:parity_k}.
Since every square sub-matrix of $G$ is invertible, we can obtain $\mathbf{m}_i = [m_{i,1}, m_{i,2}, \cdots, m_{i,k}]^T$, which completes the proof for exactly regenerating the failed systematic node.

Second, we prove that exact regeneration of the parity nodes (in the second cluster) is possible. Let 
$\mathbf{\omega}_{i,j}^{(l)}$ be the $l^{th}$ projection vector assigned for 
$N(2,j)$ in repairing $N(2,i)$. Similarly, let  $\mathbf{\omega}_{i,j}$ be the projection vector assigned for 
$N(1,j)$ in repairing $N(2,i)$. 
Assume that the parity node $N(2,i)$ fails, which contains $\mathbf{p}_i=[p_{i,1},p_{i,2},\cdots,p_{i,k}]^T$.
Then, node $N(2,j)$ transmits $\beta_I = k$ symbols $\{\mathbf{p}_j^T \mathbf{\omega}_{i,j}^{(l)}\}_{l=1}^k$, while node $N(1,j)$ transmits $\beta_c = 1$ symbol $\mathbf{m}_j^T \mathbf{\omega}_{i,j}$. For simplicity, we set 
$\mathbf{\omega}_{i,j}^{(l)}=\mathbf{e}_l$ and $\mathbf{\omega}_{i,j}=\mathbf{e}_k$.
This means that node $N(2,j)$ transmits $k$ symbols $\mathbf{p}_j=[p_{j,1},p_{j,2},\cdots,p_{j,k}]^T$ it contains, while $N(1,j)$ transmits the last symbol it contains, i.e., the symbol $m_{j,k}$.
Thus, the newcomer node for regenerating parity node $N(2,i)$ obtains the following information
\begin{equation}
P_i \coloneqq \{p_{j,s}: j \in [k] \setminus \{i\}, s \in [k] \} \cup \{m_{j,k}\}_{j=1}^k 
\end{equation}
We show how the newcomer node regenerates $\mathbf{p}_i = [p_{i,1}, p_{i,2}, \cdots, p_{i,k}]^T$ using the information $P_i$.
Among $k^2$ parity symbols in (\ref{Eqn:parity_message_relation}),
$k(k-1)$ parity symbols received by the newcomer node can be expressed as
\begin{equation}\label{Eqn:parity_k^2}
\begin{bmatrix}
\mathbf{p}_1\\
\vdots\\
\mathbf{p}_{i-1}\\
\mathbf{p}_{i+1}\\
\vdots \\
\mathbf{p}_{k}
\end{bmatrix}
=
\begin{bmatrix}
G_1^{(1)} & \cdots & G_1^{(k)}\\
\vdots & \ddots & \vdots \\
G_{i-1}^{(1)} & \cdots & G_{i-1}^{(k)}\\
G_{i+1}^{(1)} & \cdots & G_{i+1}^{(k)}\\
\vdots & \ddots & \vdots \\
G_{k}^{(1)} & \cdots & G_{k}^{(k)}\\
\end{bmatrix}
\begin{bmatrix}
\mathbf{m}_1\\
\mathbf{m}_2\\
\vdots\\
\mathbf{m}_k
\end{bmatrix}
= G'\mathbf{m},
\end{equation}
where $G_i^{(j)}$ is defined in Construction \ref{Construct:MSR_for_2k_k_2}.
Note that $G'$ is a $k(k-1) \times k^2$ matrix, which is generated by removing $l^{th}$ rows from $G$, for $l \in \{(i-1)k+1, (i-1)k+2, \cdots, ik\}$.
Since we know the values of $k$ message symbols $\{m_{j,k}\}_{j=1}^k$ and the entries of $G$ matrix, subtracting constant known values from (\ref{Eqn:parity_k^2}) results in
\begin{equation}\label{Eqn:parity_k^2_reduced}
G'' \mathbf{m}',
\end{equation}
where $G''$ is generated by removing $l^{th}$ columns from $G'$ for $l \in \{k, 2k, \cdots, k^2\}$. Similarly, $\mathbf{m}'$ is generated by removing $l^{th}$ rows from $\mathbf{m}$ for $l \in \{k, 2k, \cdots, k^2\}$. 
Thus, $G''$ is an invertible $k(k-1) \times k(k-1)$ matrix, so that 
we can obtain $\mathbf{m}'$, which contains 
\begin{equation}
\tilde{P}_i = \{m_{j,s}: j \in [k], s \in [k-1] \}.
\end{equation}
Since $P_i \cup \tilde{P}_i$ contains every message symbol $\{m_{j,s}: j,s\in [k] \}$, we can regenerate $\mathbf{p}_i = [p_{i,1}, p_{i,2}, \cdots, p_{i,k}]^T$ using (\ref{Eqn:parity_message_relation}). This completes the proof for exactly regenerating the failed parity node.

Finally, we prove that $\mathcal{M}=k^2$ message symbols can be obtained by contacting arbitrary $k$ nodes. 
In this proof, we use a slightly modified notation for representing message and parity symbols. For $j, s \in [k]$, the message symbol $m_{j,s}$ and the parity symbol $p_{j,s}$ are denoted as $m_{(j-1)k+s}$ and $p_{(j-1)k+s}$, respectively.
Then, (\ref{Eqn:parity_message_relation}) is expressed as
\begin{equation}\label{Eqn:parity_message_relation_revised}
\begin{bmatrix}
p_1 \\
p_2\\
\vdots \\
p_{k^2}
\end{bmatrix}
= G 
\begin{bmatrix}
m_1 \\
m_2\\
\vdots \\
m_{k^2}
\end{bmatrix}
\end{equation}
Suppose that the data collector (DC) contacts $e$ nodes from the $1^{st}$ cluster, and $k-e$ nodes from the $2^{nd}$ cluster, for $e \in \{0, 1, \cdots, k\}$.
Then, DC obtains $k(k-e)$ parity symbols and $ke$ message symbols. 
Since there exists total of $\mathcal{M}=k^2$ message symbols, the number of message symbols that DC cannot obtain is $\mathcal{M}-ke =k(k-e)$.
Let the parity symbols obtained by DC be $p_{i_1}, \cdots, p_{i_{k(k-e)}}$, and the message symbols not obtained by DC be $m_{j_1}, \cdots, m_{j_{k(k-e)}}$. Then, the known parities can be expressed as
\begin{equation}\label{Eqn:parity_DC}
\begin{bmatrix}
p_{i_1}\\
p_{i_2}\\
\vdots\\
p_{i_{k(k-e)}}
\end{bmatrix}
=
G'
\begin{bmatrix}
m_{1,1}\\
m_{1,2}\\
\vdots\\
m_{k,k}
\end{bmatrix},
\end{equation}
where $G'$ is a $k(k-e) \times k^2$ matrix obtained by taking $l^{th}$ rows from $G$, for $l \in \{i_t\}_{t=1}^{k(k-e)}$. Since we know $ke$ message symbols and the elements of $G$, subtracting the known constant values from (\ref{Eqn:parity_DC}) results in
\begin{equation}\label{Eqn:parity_DC_reduced}
G''
\begin{bmatrix}
m_{j_1}\\
m_{j_2}\\
\vdots\\
m_{j_{k(k-e)}}
\end{bmatrix},
\end{equation}
where $G''$ is a $k(k-e) \times k(k-e)$ matrix generated by taking the $l^{th}$ columns from $G'$ for $l \in \{j_t\}_{t=1}^{k(k-e)}$.
Since $G''$ is invertible, we obtain the unknown message symbols $\{m_{j_i}\}_{i=1}^{k(k-e)}$. This completes the proof.

\section{Proof of Propositions}
\subsection{Proof of Proposition \ref{Prop:parameter_for_small_epsilon}}\label{Section:proof_of_prop_param_small_epsilon}
From Corollary 3 of \cite{sohn2017TIT}, the MSR point for $\epsilon = 0$ is given by
\begin{equation}\label{Eqn:MSR_point_review}
(\alpha, \gamma) = 
(\frac{\mathcal{M}}{\zeta_{n_I-2}}, \frac{\mathcal{M}}{\lambda_{n_I-2}} (1- \frac{\delta_{n_I-2}}{\zeta_{n_I-2}})),
\end{equation}
where $\{\zeta_i\}, \{\lambda_i\}, \{\delta_i\}$ are defined in \cite{sohn2017TIT}. This paper does not review the explicit form of the definitions, but shows how $(\alpha,\gamma)$ looks like. 
From the proof of Lemma 5 of \cite{sohn2017TIT}, 
we have
\begin{enumerate}
	\item if $m=n_I-1$
	\begin{align}
	\lambda_{n_I-2} &= \frac{q+1}{n_I-1} \label{Eqn:lambda_divisible}\\
	\delta_{n_I-2} &= (q+1)(n_I-2) \label{Eqn:delta_divisible}\\
	\zeta_{n_I-2} &= k-q \label{Eqn:zeta_divisible}
	\end{align}
	\item else ($m=0,1,\cdots, n_I-2$)
	\begin{align}
	\lambda_{n_I-2} &= \frac{q}{n_I-1} \label{Eqn:lambda_nondivisible}\\
	\delta_{n_I-2} &= k-2q \label{Eqn:delta_nondivisible}\\
	\zeta_{n_I-2} &= k-q \label{Eqn:zeta_nondivisible}
	\end{align}
\end{enumerate}
where $q$ and $m$ are defined in (\ref{Eqn:remainder}) and (\ref{Eqn:quotient}).
When $m=n_I-1$, (\ref{Eqn:delta_divisible}) can be expressed as
\begin{align}\label{Eqn:delta_divisible_review}
\delta_{n_I-2}&=qn_I-2q+n_I-2 = qn_I-2q+m-1 \nonumber\\
&= k-2q-1,
\end{align}
where the last equality is from (\ref{Eqn:remainder}). Thus, from (\ref{Eqn:delta_divisible_review}), (\ref{Eqn:zeta_divisible}) and (\ref{Eqn:lambda_divisible}), we have
\begin{equation}\label{Eqn:relationship_MSR}
\frac{\zeta_{n_I-2}-\delta_{n_I-2}}{\lambda_{n_I-2}} = n_I-1
\end{equation}
holds for the $m=n_I-1$ case. Similarly, using (\ref{Eqn:lambda_nondivisible}), (\ref{Eqn:delta_nondivisible}) and (\ref{Eqn:zeta_nondivisible}), we can confirm that (\ref{Eqn:relationship_MSR}) holds for the $0 \leq m \leq n_I-2$ case. Inserting (\ref{Eqn:zeta_divisible}), (\ref{Eqn:zeta_nondivisible}), (\ref{Eqn:relationship_MSR}) into (\ref{Eqn:MSR_point_review}), we obtain
\begin{align}
\alpha &= \frac{\mathcal{M}}{k-q}, \\
\gamma &= \frac{\mathcal{M}}{k-q} (n_I-1)
\end{align}
Since $\gamma = (n_I-1)\beta_I$ for $\beta_c = 0$ from (\ref{Eqn:gamma}), we obtain 
\begin{equation}
\beta_I = \alpha = \mathcal{M}/(k-q),
\end{equation}
which completes the proof.

\subsection{Proof of Proposition \ref{Prop:parameter_for_large_epsilon}}\label{Section:proof_of_prop_param_large_epsilon}

We consider the $\beta_c = 1$ case without losing generality. This implies that
\begin{equation}\label{Eqn:beta_I_large_epsilon}
\beta_I = 1/\epsilon=n-k
\end{equation}
according to the definition $\epsilon=\beta_c/\beta_I$. Now, we observe the expressions for $\alpha$ and $\mathcal{M}$.
From Corollary 3 of \cite{sohn2017TIT}, the MSR point for $\epsilon=1/(n-k)$ is illustrated as
\begin{equation}\label{Eqn:MSR_point_review2}
(\alpha, \gamma) = (\frac{\mathcal{M}}{k}, \frac{\mathcal{M}}{k}\frac{1}{s_{k-1}}),
\end{equation}
where 
\begin{align}\label{Eqn:s_seq}
s_{k-1} &= \frac{(n-k)\epsilon}{(n_I-1) + \epsilon (n-n_I)} = \frac{1}{n_I-1+\frac{n-n_I}{n-k}}
\end{align}
from the definition of $\{s_i\}$ in \cite{sohn2017TIT} and the setting of $\epsilon=1/(n-k)$.
Combining (\ref{Eqn:MSR_point_review2}) and (\ref{Eqn:s_seq}) result in (\ref{Eqn:parameters_for_large_epsilon}).

Note that $\gamma$ in (\ref{Eqn:gamma}) can be expressed as
\begin{align}
\gamma &= (n-n_I)\beta_c + (n_I-1)\beta_I = (n-n_I) + (n_I-1)(n-k) , \label{Eqn:gamma_epsilon}
\end{align}
where the last equality holds due to \eqref{Eqn:beta_I_large_epsilon}.
Combining \eqref{Eqn:parameters_for_large_epsilon} and (\ref{Eqn:gamma_epsilon}), we obtain
\begin{equation*}
\gamma = \frac{\mathcal{M}}{k}\frac{\gamma}{n-k},
\end{equation*}
which result in 
\begin{equation}\label{Eqn:Capacity_MSR_point}
\mathcal{M} = k(n-k).
\end{equation}
Using $\alpha=\mathcal{M}/k$ in (\ref{Eqn:MSR_point_review2}),
we have
\begin{equation}
\alpha = n-k.
\end{equation}
This completes the proof.



\bibliographystyle{IEEEtran}
\bibliography{IEEEabrv,IEEE_ISIT2018}






\IEEEtriggeratref{3}

\end{document}